\documentclass{llncs}

\usepackage[utf8]{inputenc}
\usepackage{amsmath}
\usepackage{amssymb}
\usepackage[inline]{enumitem}
\usepackage{mathtools}
\usepackage{microtype}
\usepackage{todonotes}
\usepackage{csquotes}
\usepackage{thm-restate}

\binoppenalty=\maxdimen
\relpenalty=\maxdimen

\DeclareMathOperator{\corr}{corr}
\DeclareMathOperator{\skel}{skel}

\DeclarePairedDelimiterX\set[1]\lbrace\rbrace{#1}

\title{An SPQR-Tree-Like Embedding Representation for Upward Planarity}
\author{Guido Br\"uckner\inst{1} \and
        Markus Himmel\inst{1} \and
        Ignaz Rutter\inst{2}}

\institute{Karlsruhe Institute of Technology \and
  University of Passau}

\begin{document}

    \maketitle

    \begin{abstract}
        The SPQR-tree is a data structure that compactly represents all planar embeddings of a biconnected planar graph.
        It plays a key role in constrained planarity testing.

        We develop a similar data structure, called the UP-tree, that compactly represents all upward planar embeddings of a biconnected single-source directed graph.
        We demonstrate the usefulness of the UP-tree by solving the upward planar embedding extension problem for biconnected single-source directed graphs.
    \end{abstract}

    \section{Introduction}

        A natural extension of planarity to directed graphs (digraphs) is to
        consider planar drawings where each edge is drawn as a~$y$-monotone
        curve.  Such drawings are called \emph{upward planar}, and a graph
        admitting an upward planar drawing is \emph{upward planar}.  A
        planar (combinatorial) embedding~$\mathcal E$ of a graph~$G$ is an \emph{upward planar
          embedding} if~$G$ has an upward planar drawing whose (combinatorial) embedding
        is~$\mathcal E$.
        Whereas undirected graphes can be tested for planarity in linear time,
        upward planarity testing is \textsf{NP}-complete in general, though there are efficient algorithms for graphs with a single source~\cite{hutton1996upward,bertolazzi1998optimal} and graphs with a fixed embedding~\cite{bertolazzi1994upward}.
        In the special case of \emph{$st$-graphs}, i.e., graphs with a single
        source~$s$ and a single sink~$t$ with~$s$ and~$t$ on the same face, 
        planar embeddings are the same as the upward planar
        embeddings~\cite{platt1976planar}, and hence upward
        planarity and planarity are equivalent.

        A related but different planarity notion for digraphs is \emph{level planarity}, where the vertices of the graph have fixed levels that correspond to horizontal lines in the drawing.
        The task is to order the vertices on each level so that the drawing is planar.
        Level planarity can be tested in linear time~\cite{jl-lpeilt-99} by a quite involved algorithm, or in quadratic time by several simpler algorithms~\cite{brs-lptec-18,rsbhkmsc-asfopolg-01,fpss-htmdalp-13}.

        In a constrained embedding problem, one seeks a planar embedding
        of a given graph that satisfies additional constraints.  Typical
        examples are simultaneous embeddings with fixed
        edges~\cite{bkr-sepg-14}, cluster planarity~\cite{fce-pcg-95},
        constraints on the face sizes~\cite{djkr-pesuf-14,dr-aafcp-18}, optimizing the depth of the
        embedding~\cite{adp-fmdep-11} and optimizing the bends in an orthogonal
        drawing~\cite{blr-ogdie-16,brw-oogdc-16,dlp-bmodq-18}.  One of the
        most prominent examples of the last years is the \emph{partial
          drawing extension} problem, which asks whether a given drawing
        of a subgraph can be extended to a planar drawing of the whole
        graph.  The \emph{partial embedding extension}
        problem is strongly related, here the input is a planar embedding of a subgraph and
        the question is whether it can be extended to a planar embedding
        of the whole graph.  For undirected planar graphs the two problems
        are equivalent and can be solved in linear
        time~\cite{adfjk-tppeg-15,jkr-ktppe-13}.

        One of the key tools for all of these applications is the
        SPQR-tree, which compactly represents all planar embeddings of a
        biconnected planar graph~$G$ and breaks down the complicated task
        of choosing a planar embedding of~$G$ into simpler independent
        embedding choices of its triconnected components~\cite{ml-ascopcg-37,t-cig-66,hopcroft1973dividing,dbt-ipt-89,dbt-olgawst-90,dt-omtcs-96}.  In fact, these embeddings are either uniquely
        determined up to reversal, or they consist in arbitrarily choosing a
        permutation of parallel edges between two pole vertices.  The
        common approach for attacking the above-mentioned constrained
        embedding problems is to project the constraints on the global
        embedding to local constraints on the skeleton embeddings that can
        then be satisfied by consistent local choices.  While the
        implementation details are often highly technical and non-trivial, the
        approach has proven to be extremely successful.

        In comparison, relatively little is known about constrained
        planarity problems for planarity notions of digraphs.
        Br\"uckner and Rutter~\cite{br-pclp-17} study the problem of
        extending a given partial drawing of a level graph and Da Lozzo et
        al.~\cite{ddf-eupgd-19} study the same question for upward
        planarity.  In general, extending a given partial upward planar
        drawing requires to determine an upward planar embedding that (i)
        extends the embedding of the partial drawing, and (ii) admits a
        drawing that extends the given drawing.  Here step (i) requires
        solving the embedding extension problem but with additional
        constraints that ensure that a drawing extension is feasible.  It
        is worth noting that for upward planarity the embedding extension
        problem and the drawing extension problem are distinct; Da Lozzo
        et al. show that, generally, even if an upward planar embedding of
        the whole graph is given, it is NP-complete to decide whether it
        can be drawn such that it extends a given partial
        drawing~\cite[Theorem 2]{ddf-eupgd-19}.  On the positive side,
        they present tractability results for directed paths and cycles
        with a given upward planar embedding, and for~$st$-graphs.  The
        restriction to~$st$-graphs allows a relatively simple
        characterization of the upward planar embeddings that extend the
        given partial drawings~\cite[Lemma 6]{ddf-eupgd-19}, which yields
        an~$O(n \log n)$-time algorithm for step (ii).  For step (i), Da
        Lozzo et al.~exploit the fact that for~$st$-graphs, the choice of an
        upward planar embedding is equivalent to choosing a planar
        embedding, and hence the SPQR-tree allows to efficiently search
        for an upward planar embedding satisfying the additional
        constraints required by condition (ii).

        In this paper, we seek to generalize the approach of Da Lozzo et
        al.~to biconnected single-source graphs.  The key difficulty in
        this case is that neither do we have access to all the upward
        planar embeddings of such graphs, nor is it known what the
        necessary and sufficient conditions are for an upward planar embedding
        to admit a drawing that extends a given subdrawing.

        \paragraph{Contribution and Outline.} We construct a novel SPQR-tree-like embedding representation,
        called the UP-tree, that represents exactly the upward planar
        embeddings of a biconnected single-source graph. As in SPQR-trees, the embedding
        choices in the UP-tree are broken down into independent embedding choices of
        skeleton graphs that are either unique up to reversal or allow to
        arbitrarily permute parallel edges between two poles.
        As such, UP-trees can take the role of SPQR-trees for constrained embedding problems in upward planarity, making them a powerful tool with a broad range of applications.
        We demonstrate this by giving an quadratic-time algorithm for the upward planar embedding extension problem for biconnected single-source graphs.

        After introducing some preliminaries in Section~\ref{sec:preliminaries}, we review the results on decomposing upward planar single-source digraphs due to Hutton and Lubiw~\cite{hutton1996upward} in Section~\ref{sec:decomposition-trees-and-upward-planar-embeddings}.
        We proceed to extend this idea from a single decomposition to decomposition trees.
        Proofs of statements marked with a star ($\star)$ can be found in the appendix.
        In Section~\ref{sec:up-trees}, we define the UP-tree and in Section~\ref{sec:partial-upward-embedding} we use it to solve the partial upward embedding extension problem.

    \section{Preliminaries}
    \label{sec:preliminaries}

        Let~$G = (V, E)$ be a connected simple undirected graph.
        A \emph{cutvertex} of~$G$ is a vertex whose removal disconnects~$G$.
        We say that~$G$ is \emph{biconnected} if it has no cutvertex.
        We say that~$\{u, v\}$ is a \emph{cutpair} if there are connected subgraphs~$H_1, H_2$ of~$G$ with~$H_1 \cup H_2 = G$ and~$H_1 \cap H_2 = \{u, v\}$.
        If a graph has no cutpair it is \emph{triconnected}.

        \paragraph{Decomposition Trees.}
        Assume that~$G$ is biconnected.
        A \emph{decomposition} along a cutpair~$\{u, v\}$ of~$G$ is defined as follows.
        Let~$\mu_1, \mu_2$ be two nodes connected by an undirected arc~$(\mu_1, \mu_2)$.
        Node~$\mu_i$ is equipped with a multigraph~$H_i \cup \{(u, v)\}$ called its \emph{skeleton} denoted by~$\skel(\mu_i)$.
        The newly added edge~$(u, v)$ is a \emph{virtual edge} and corresponds to~$\mu_2$ in~$\mu_1$ and to~$\mu_1$ in~$\mu_2$, respectively.
        This is formalized as functions~$\corr_{\mu_1}\colon (u, v) \mapsto \mu_2$ and~$\corr_{\mu_2}\colon (u, v) \mapsto \mu_1$.
        If there exists a cutpair~$\{u', v'\}$ in~$\skel(\mu_i)$ and we may once again decompose along that cutpair.
        By repeating this process we obtain an unrooted \emph{decomposition tree}~$\mathcal T$.

        Let~$a = (\mu, \nu)$ be an arc of~$\mathcal T$.
        Then~$\skel(\mu)$ and~$\skel(\nu)$ share two vertices~$u, v$ and the existence of~$a$ can be traced back to a decomposition along~$u, v$.
        We then say that the \emph{poles} of~$\mu$ in~$\nu$ are~$u$ and~$v$.
        When~$\nu$ is clear from the context we also simply refer to~$u$ and~$v$ as the poles of~$\mu$.

        A decomposition can be reverted by contracting an arc~$(\mu, \nu)$ of~$\mathcal T$ and merging the skeletons of~$\mu$ and~$\nu$.
        To merge~$\skel(\mu)$ and~$\skel(\nu)$, remove from~$\skel(\mu)$ the virtual edge~$e$ with~$\corr_\mu(e) = \nu$ and from~$\skel(\nu)$ the virtual edge~$e'$ with~$\corr_\nu(e') = \mu$ and set the union of these two graphs as the skeleton of the node obtained by contracting~$(\mu, \nu)$ in~$\mathcal T$.
        This is a \emph{composition} along~$(\mu, \nu)$.

        Consider an arc~$a = (\mu, \nu)$ of~$\mathcal T$.
        Removing~$a$ from~$\mathcal T$ separates~$\mathcal T$ into two subtrees~$\mathcal T_\mu$ and~$\mathcal T_\nu$ containing~$\mu$ and~$\nu$, respectively.
        Define the \emph{pertinent graph} of~$\mu$ in~$\nu$ as the skeleton of the single node obtained by contracting all arcs in~$\mathcal T_\mu$.
        Again, when~$\nu$ is clear from the context we simply refer to this graph as the pertinent graph of~$\mu$ and denote it by~$G(\mu)$.

        \paragraph{Rooted Decomposition Trees and Planar Embeddings.}
        Throughout this paper let an \emph{embedding} of a graph denote a rotation system together with an outer face.
        Decomposition trees can be used to decompose not only a graph, but also an embedding of it.
        Consider a biconnected graph~$G$ together with a planar embedding~$\mathcal E$.
        Let~$e^\star$ be an edge of~$G$ incident to the outer face of~$\mathcal E$.
        Further, let~$\mathcal T$ be a decomposition tree of~$G$ rooted at a node whose skeleton contains~$e^\star$.
        Equip the skeleton of each node~$\mu$ of~$\mathcal T$ with an embedding as follows.
        The embedding of~$\skel(\mu)$ is obtained from~$\mathcal E$ by contracting for each virtual edge~$(u, v)$ of~$\skel(\mu)$ the pertinent graph~$G(\corr_\mu(u, v))$ into a single edge.
        These embeddings of the skeletons of the nodes of~$\mathcal T$ are referred to as a \emph{configuration}.
        The fact that~$e^\star$ is incident to the outer face gives two properties.
        First, the edge~$e^\star$ lies on the outer face of the skeleton of the root node of~$\mathcal T$.
        Second, every non-root node~$\nu$ of~$\mathcal T$ has some parent node~$\mu$ and the virtual edge~$e$ with~$\corr_\nu(e) = \mu$ lies on the outer face of~$\skel(\nu)$.
        We extend our notion of a \emph{configuration} to any set of embeddings of the skeletons of the nodes of~$\mathcal T$ that fulfills these two properties.

        Recall that decomposition trees allow for (graph) composition along arcs.
        We can also compose embeddings.
        When contracting an arc~$(\mu, \nu)$, we merge~$\skel(\mu)$ and~$\skel(\nu)$ as described above.
        Obtain the embedding of the merged skeleton by replacing the occurrences of the virtual edge in the rotation system around its poles by the appropriate rotation system in the embedding of the other skeleton.
        This means that~$G$ together with a planar embedding can be decomposed into a decomposition tree~$\mathcal T$ together with a configuration.
        And symmetrically,~$\mathcal T$ together with any configuration can be composed into a planar embedding of~$G$.

        \paragraph{SPQR-Trees.}
        As described in the previous paragraph, decomposition trees separate independent choices in finding planar embeddings of a graph.
        We may either choose an embedding of the entire graph, which is generally very complex, or we may decompose the graph into smaller skeletons, independently choose embeddings of these skeletons and compose them into an embedding of the entire graph.
        In this sense decomposition trees implement a tradeoff between making few complex choices or many simple choices.

        The \emph{SPQR-tree} is a decomposition tree that makes this tradeoff in favor of many simple choices.
        SPQR-trees have four kinds of nodes, all of whose skeletons offer only few and well-structured embedding choices.
        \begin{enumerate*}[label=(\roman*)]
            \item R-nodes are nodes whose skeleton is triconnected.
                  Such skeletons have a unique planar embedding up to flipping.
            \item S-nodes are nodes whose skeleton is a simple cycle.
                  Such skeletons offer no embedding choice (recall that the outer face is fixed by the rooting).
                  Adjacent S-nodes are contracted into one larger S-node, i.e., an S-node whose skeleton is a larger simple cycle.
                  This means that in SPQR-trees no two S-nodes are adjacent.
            \item P-nodes are nodes whose skeleton is a multigraph that consists of two vertices connected by three or more edges.
                  The order of these edges may be arbitrarily permuted.
                  Again, adjacent P-nodes are contracted into one larger P-node, i.e., no two P-nodes are adjacent.
            \item Q-nodes are nodes whose skeleton consists of two vertices connected by two edges, namely one virtual edge and one non-virtual edge.
                  They offer no embedding choice.
                  Note that only the skeletons of Q-nodes contain non-virtual edges.
        \end{enumerate*}
        See Fig.~\ref{fig:example}~(a) and~(b) for a graph and its SPQR-tree decomposition.

    \section{Decomposition Trees and Upward Planar Embeddings}
    \label{sec:decomposition-trees-and-upward-planar-embeddings}

        Recall from the previous section that for biconnected graphs we can decompose any planar embedding into planar embeddings of the skeletons of a decomposition tree; and symmetrically, we can compose a planar embedding of the whole graph from planar embeddings of the skeletons.
        In this section we find a similar relationship between upward planar embeddings of a biconnected single-source digraph~$G$ and upward planar embeddings of the skeletons of a suitably-defined decomposition tree of~$G$.

        \subsection{Decompositions and Upward Planar Embeddings}
        \label{ssec:decompositions-and-upward-planar-embeddings}

        In this section we review the decomposition result of Hutton and Lubiw and formulate the interface between their result and our results.

        Let~$G$ be a biconnected single-source digraph together with an upward planar embedding~$\mathcal E$.
        Further, let~$e^\star$ denote the edge around the source of~$G$ that is leftmost in~$\mathcal E$.
        Now let~$H_1, H_2$ be two subgraphs of~$G$ with
        \begin{enumerate*}[label=(\roman*)]
            \item $H_1 \cup H_2 = G$,
            \item $H_1 \cap H_2 = \{u, v\}$,
            \item $e^\star \in H_1$ and
            \item $H_1 \setminus \{u, v\}$ or~$H_2 \setminus \{u, v\}$ is connected.
        \end{enumerate*}

        Hutton and Lubiw construct two graphs~$H_1'$ and~$H_2'$ from~$H_1$ and~$H_2$ by including one of the \emph{markers} shown in Fig.~\ref{fig:markers}.
        \begin{figure}[t]
            \centering
            \includegraphics{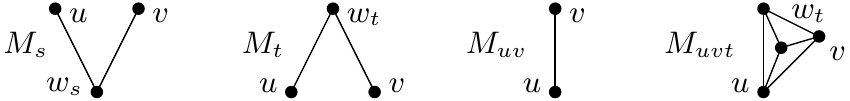}
            \caption{
                The four markers used by Hutton and Lubiw.
                The markers are digraphs; in the figure, all edges are directed upward.
            }
            \label{fig:markers}
        \end{figure}
        Markers are simple digraphs with two vertices~$u, v$ that connect the marker to the remaining graph.
        The marker in~$H_1'$ is designed to represent~$H_2$ and the marker in~$H_2'$ is designed to represent~$H_1$.
        If there exists a directed path from~$u$ to~$v$ we say that~$u$ \emph{dominates}~$v$ and write~$u < v$ for short.
        Otherwise~$u$ and~$v$ are \emph{incomparable}.
        The vertex~$v$ is a \emph{source} if it has no incoming edges in~$G$, a \emph{sink} if it has no outgoing edges in~$G$ and an \emph{internal} vertex if it has both incoming and outgoing edges in~$G$.
        Markers are determined based on whether~$u < v$ and whether~$v$ is a source, sink or internal vertex in~$H_1$ and~$H_2$:
        If~$u$ and~$v$ are incomparable in~$G$, set~$H_1' = H_1 \cup M_t$ and~$H_2' = H_2 \cup M_s$.
        Otherwise, assume~$u < v$.
        Define~$H_1'$ as follows.
        If~$v$ is a source in~$H_2$ set~$H_1' = H_1 \cup M_t$.
        If~$v$ is a sink in~$H_2$ set~$H_1' = H_1 \cup M_{uv}$.
        Otherwise~$v$ is an internal vertex in~$H_2$ and we set~$H_1' = H_1 \cup M_{uvt}$.
        Define~$H_2'$ as follows.
        If~$v$ is a source in~$H_1$ set~$H_2' = H_2 \cup M_t$, otherwise set~$H_2' = H_2 \cup M_{uv}$.
        See Fig.~\ref{fig:marker-replacement} in the appendix for example decompositions.

        Recall that decomposition trees of planar graphs allow for (de-)composition of planar embeddings.
        Hutton and Lubiw provide a similar property for the graphs~$G, H_1'$ and~$H_2'$.

        \begin{restatable}[$\star$, implicit in \cite{hutton1996upward}]{theorem}{theoremUpwardPlanarConfigurations}
            Let~$\mathcal E$ be an upward planar embedding of~$G$ with~$e^\star$ as the leftmost edge around~$s$.
            Then~$\mathcal E$ induces upward planar embeddings~$\mathcal F_1, \mathcal F_2$ of~$H_1', H_2'$, respectively with the following properties.
            In~$\mathcal F_1$,~$e^\star$ is the leftmost edge around the source of~$H_1'$.
            In~$\mathcal F_2$, the edges of the marker are leftmost around the source of~$H_2'$.
            Conversely, if~$\mathcal{F}_1$ and~$\mathcal{F}_2$ are upward planar embeddings of~$H_1'$ and~$H_2'$
            such that~$e^\star$ is the leftmost edge around the source of~$H_1'$ and the edges
            of the marker are leftmost around the source of~$H_2'$, then the composition
            of these embeddings is upward planar.
            \label{thm:upward-planar-configurations}
      \end{restatable}

        Hutton and Lubiw do not explicitly state Theorem~\ref{thm:upward-planar-configurations}.
        Instead, Theorems 6.5, 6.7, 6.8 and 6.9 in \cite{hutton1996upward} discuss the same
        situation as Theorem~\ref{thm:upward-planar-configurations}, but are only
        concerned with upward planarity, not with the embeddings involved.
        See the appendix for a detailed discussion.

        \subsection{Decomposition Trees and Upward Planar Embeddings}
        \label{ssec:decomposition-trees-and-upward-planar-embeddings}

        The approach of Hutton and Lubiw is to decompose a single-source digraph~$G$ into two smaller single-source digraphs~$G_1, G_2$ and use Theorem~\ref{thm:upward-planar-configurations} to translate upward-planarity testing of~$G$ to upward-planarity testing of two smaller instances~$H_1', H_2'$.
        Observe that the markers and the replacement rules are defined so that both~$H_1'$ and~$H_2'$ are single-source digraphs.
        This means that~$H_1'$ and~$H_2'$ can be recursively decomposed.
        Note that in the context of connectivity markers are treated simply as edges, i.e., markers are not decomposed further.
        When a graph cannot be further decomposed it is triconnected and therefore has a unique planar embedding which can be tested for upward planarity in linear time using the algorithm of Bertolazzi et al.~\cite{bertolazzi1998optimal}.
        In the context of upward planarity testing the full marker graph is considered.
        Upward planar embeddings of~$H_1'$ and~$H_2'$ can then be composed to an upward planar embedding of~$G$.
        In the context of embedding composition markers are again treated simply as edges.
        In particular, it does not matter whether the clockwise order of the edges incident to~$u$ in~$M_{uvt}$ is~$(u, v), (u, x), (u, w_t)$ or~$(u, w_t), (u, x), (u, v)$.

        We use a different approach.
        Instead of testing~$H_1'$ and~$H_2'$ for upward-planarity separately, we manage them as the skeletons of two nodes in a decomposition tree~$\mathcal T$.
        Note that Theorem~\ref{thm:upward-planar-configurations} requires~$H_1 \setminus \{u, v\}$ or~$H_2 \setminus \{u, v\}$ to be connected.
        We call such a decomposition \emph{maximal}.
        We then decompose these skeletons further, which grows the decomposition tree.
        A \emph{maximal-decomposition tree} is a decomposition tree obtained by performing only maximal decompositions.
        A \emph{configuration} equips the skeleton of each node in the tree with an upward planar embedding.
        In this embedding,~$e^\star$ or the marker that represents the component that contains~$e^\star$ must be incident to the outer face and leftmost around the source of the skeleton.
        See Fig.~\ref{fig:example}~(c) for an example of a maximal-decomposition tree.
        Applying Theorem~\ref{thm:upward-planar-configurations} at each decomposition step gives the following.

        \begin{theorem}
            Let~$G$ be a biconnected graph with a single source~$s$, let~$e^\star$ be an edge of~$G$ incident to~$s$ and let~$\mathcal T$ denote a maximal-decomposition tree of~$G$.
            Then the upward-planar embeddings of~$G$ in which~$e^\star$ is the leftmost edge around~$s$ correspond bijectively to the configurations of~$\mathcal T$.
            \label{thm:upward-planar-decomposition-tree}
        \end{theorem}

        We could use Theorem~\ref{thm:upward-planar-decomposition-tree} directly to represent all upward planar embeddings of a graph.
        But we also show that decomposition trees are uniquely defined by the decompositions that are executed, but not by the order of these decompositions.
        This means that just like we can talk about \emph{the} SPQR-tree decomposition for a graph we will be able to talk about \emph{the} UP-tree decomposition.
        The benefit of this is that we can use a UP-tree decomposition to determine that some constrained representation problem has no solution without having to consider other conceivable UP-tree decompositions.

        To prove uniqueness, we show that the order of the decompositions is irrelevant.
        We then apply the decompositions as defined by the SPQR-tree decomposition, which is unique, and obtain the unique UP-tree decomposition.
        To this end, we prove that the marker replacement rules do not depend on the order of the decompositions.
        Recall that the marker replacement rules depend on vertex dominance and the local neighborhood of certain vertices.
        We prove Lemma~\ref{lem:decompositions-preserve-vertex-dominance}, which states that decompositions preserve vertex dominance and Lemma~\ref{lem:decompositions-preserve-neighborhood}, which states that decompositions preserve the local neighborhood of certain vertices.

        \begin{restatable}[$\star$]{lemma}{lemmaDecompositionsPreserveVertexDominance}
            Let~$G$ be a biconnected single-source digraph and let~$H_1', H_2'$ denote the result of decomposing along a cutpair~$\{u, v\}$ of~$G$.
            For~$i = 1, 2$ and any two vertices~$x, y$ in~$H_i'$ it is~$x < y$ in~$H_i'$ if and only if~$x < y$ in~$G$.
            \label{lem:decompositions-preserve-vertex-dominance}
          \end{restatable}

        \begin{restatable}[$\star$]{lemma}{lemmaDecompositionsPreserveNeighborhood}
            Let~$G$ be a biconnected single-source digraph and let~$H_1', H_2'$ denote the result of decomposing along a cutpair~$\{u , v\}$ of~$G$.
            For~$i = 1, 2$ let~$\{x, y\}$ denote a cutpair of~$H_i'$ that separates~$H_i'$ into~$F_1$ and~$F_2$ and~$G$ into~$D_1$ and~$D_2$.
            Then~$y$ is a source in~$F_1$ if and only if~$y$ is a source in~$D_1$.
            Moreover,~$y$ is a source, sink or internal vertex in~$F_2$ if and only if~$y$ is a source, sink or internal vertex in~$D_2$, respectively.
            \label{lem:decompositions-preserve-neighborhood}
          \end{restatable}

        \noindent
        Lemmas~\ref{lem:decompositions-preserve-vertex-dominance} and~\ref{lem:decompositions-preserve-neighborhood} immediately give the following.

        \begin{lemma}
            Let~$G$ be a biconnected graph with a single source~$s$, let~$e^\star$ be an edge of~$G$ incident to~$s$ and let~$\mathcal T$ denote a decomposition tree of~$G$.
            Then~$\mathcal T$ relative to~$e^\star$ is uniquely defined by the decompositions regardless of their order.
        \end{lemma}

        A configuration of~$\mathcal T$ can be computed as follows.
        Recall that all skeletons are single-source digraphs.
        We may therefore run the algorithm due to Bertolazzi et al.~\cite{bertolazzi1998optimal} on each skeleton.
        Observe that in a configuration of~$\mathcal T$ relative to~$e^\star$ the skeleton of each node~$\mu$ of~$\mathcal T$ must be embedded so that~$e^\star$ or the marker that corresponds to the component that contains~$e^\star$ must appear leftmost around the source of~$\skel(\mu)$.
        We can enforce this by rooting the decomposition tree constructed by the algorithm of Bertolazzi et al.~at the Q-node corresponding to~$e^\star$ or an edge of the marker that corresponds to the component that contains~$e^\star$.

        \section{UP-Trees}
        \label{sec:up-trees}

            We are ready to construct the UP-tree, a maximal-decomposition tree designed to mimic the SPQR-tree.
            Let~$G$ be a biconnected directed single-source graph.
            The base of the construction is the decomposition tree obtained by performing the same set of decompositions as in the construction of the SPQR-tree decomposition of the underlying undirected graph of~$G$.
            We then perform two additional steps.
            The first step is to split P-nodes into chains of smaller nodes.
            The second step is to determine whether skeletons of R-nodes can be reversed and to contract some arcs of the decomposition tree.
            In both steps, we reason about upward planarity of fixed embeddings with the following lemma due to Bertolazzi et al.~\cite{bertolazzi1998optimal}.

            Let~$G$ be a biconnected single-source graph together with a planar embedding.
            The \emph{face-sink graph}~$F$ of~$G$ has the vertices and faces of~$G$ as its vertices.
            It contains an undirected edge~$\{f, v\}$ if~$f$ is a face of~$G$ and~$v$ is a vertex of~$G$ that is incident to~$f$ and both edges incident to~$v$ and~$f$ are directed towards~$v$.
            The following lemma implies a linear-time algorithm that tests an embedding for upward planarity and outputs for each face whether it can be the outer face.
            \begin{lemma}[{\cite[Theorem 1]{bertolazzi1998optimal}}]
                Let~$G$ be an embedded planar single-source digraph and let~$h$ be a face of~$G$.
                Graph~$G$ has an upward planar drawing that preserves the embedding with outer face~$h$ if and only if all of the following is true:
                \begin{enumerate*}[label=(\roman*)]
                    \item graph~$F$ is a forest
                    \item there is exactly one tree~$T$ of~$F$ with no internal vertices of~$G$, while the remaining trees have exactly one internal vertex;
                    \item $h$ is in tree~$T$; and
                    \item the source of~$G$ is in the boundary of~$h$.
                \end{enumerate*}
                \label{lem:upward-planarity-fixed-embedding}
            \end{lemma}

            \subsection{P-Node Splits}
			\label{sec:p-node-splits}

                In SPQR-trees, the edges of P-nodes may be arbitrarily permuted.
                In decomposition trees for upward planar graphs there are stricter rules for the ordering of the markers in P-nodes.
                In this section, we determine these rules and find that by breaking up the P-nodes into chains of smaller nodes we obtain a decomposition tree for upward planarity whose P-nodes exhibit the same behavior as in SPQR-trees, i.e., their edges may be arbitrarily permuted.
                The idea is that certain kinds of markers must appear consecutively.

                First, we argue that all~$M_{uv}$ markers must appear consecutively.
                To see this, note that if~$M_s$ appears between two~$M_{uv}$ markers then the outer face is not incident to the source of the skeleton, which is vertex~$w_s$ of~$M_s$.
                If a marker~$M$ with~$M = M_t$ or~$M = M_{uvt}$ appears between two~$M_{uv}$ markers then the face incident to~$w_t$ of~$M$ and a marker~$M_{uv}$ is not connected to the outer face and not connected to an internal vertex.
                In all cases the conditions from Lemma~\ref{lem:upward-planarity-fixed-embedding} are violated.

                Moreover, all~$M_{uv}$ and~$M_{uvt}$ markers must appear consecutively.
                To see this, note that if~$M_t$ appears between two markers~$M_{uv}$ or~$M_{uvt}$ the vertex~$w_t$ of~$M_t$ cannot be connected to an internal vertex or the outer face and apply Lemma~\ref{lem:upward-planarity-fixed-embedding}.

                These observations motivate the following restructuring of P-nodes.
                Let~$\lambda$ denote a P-node obtained from the SPQR-tree.
                The \emph{parent marker} in~$\skel(\lambda)$ is the marker that corresponds to the parent node of~$\lambda$.
                If the parent marker in~$\skel(\lambda)$ is~$M_s$ all other markers must be~$M_t$.
                In this case these markers can already be arbitrarily permuted and nothing further needs to be shown.
                Otherwise the parent marker is~$M_t$ or~$M_u$ (recall that by definition of~$H_2'$ the parent marker is not~$M_{uvt}$).
                See Fig.~\ref{fig:p-node-splits}~(a) where the parent marker is~$M_t$ (the case for~$M_u$ is similar).
                \begin{figure}[t]
                    \centering
                    \includegraphics{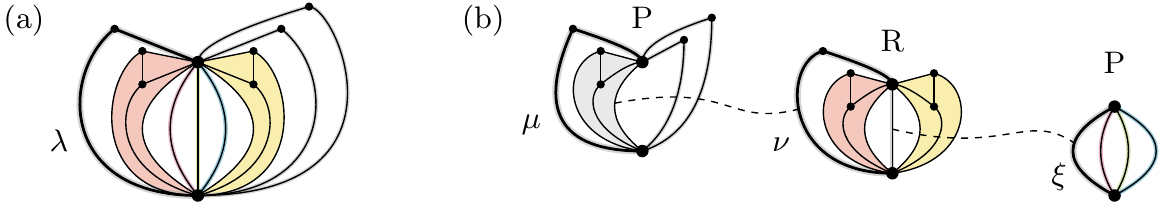}
                    \caption{
                        Splitting a P-node~$\lambda$ obtained from the SPQR-tree~(a) into a chain of smaller nodes~$\mu, \nu, \xi$~(b).
                        The bold marker represents the component that contains the edge~$e^\star$.
                    }
                    \label{fig:p-node-splits}
                \end{figure}
                Because all~$M_{uv}$ and~$M_{uvt}$ markers must appear consecutively, we create a new P-node~$\mu$ that contains the parent marker of~$\skel(\lambda)$, all~$M_t$ markers of~$\skel(\lambda)$ and a single~$M_{uvt}$ marker to represent all~$M_{uv}$ and~$M_{uvt}$ markers of~$\skel(\lambda)$.
                This marker corresponds to a new node~$\nu$ that contains all~$M_{uvt}$ markers of~$\skel(\lambda)$ and---because all~$M_{uv}$ markers must appear consecutively---a single~$M_{uv}$ marker.
                This marker corresponds to a new node~$\xi$ that contains all~$M_{uv}$ markers of~$\skel(\lambda)$.
                If~$\skel(\lambda)$ contains no~$M_{uvt}$ marker we can include a~$M_{uv}$ marker instead of a~$M_{uvt}$ marker in~$\skel(\mu)$ and connect it directly to~$\xi$, the node~$\nu$ can then be omitted.

                The new node~$\mu$ has the property that its markers can be arbitrarily permuted, i.e., it is a P-node.
                Observe that there can be at most two~$M_{uvt}$ markers in~$\skel(\lambda)$.
                This means that~$\skel(\nu)$ has at most four markers and its embedding is fixed up to reversal, i.e., it is an R-node.
                Finally, the new node~$\xi$ also has the property that its markers can be arbitrarily permuted, i.e., it is also a P-node.
                See Fig.~\ref{fig:p-node-splits}~(b) and Fig.~\ref{fig:example}~(c) and~(d) for a larger example.
                We conclude the following.

                \begin{lemma}
                    Let~$G$ be a biconnected digraph with a single source~$s$ and let~$e^\star$ denote an edge incident to~$s$.
                    There exists a decomposition tree~$\mathcal T$ that
                    \begin{enumerate*}[label=(\roman*)]
                        \item represents all upward planar embeddings of~$G$ in which~$e^\star$ is the leftmost edge around~$s$, and
                        \item the children of all P-nodes in~$\mathcal T$ can be arbitrarily permuted.
                    \end{enumerate*}
                \end{lemma}

            \subsection{Arc Contractions}
			\label{sec:arc-contractions}

                Recall that in SPQR-trees the skeletons of R-nodes are triconnected, i.e., their planar embedding is fixed up to reversal.
                So, every R-node offers one degree of freedom, namely, whether it has some reference embedding or the reversal thereof.
                In this section we alter our decomposition tree so that it has this same property.

                By definition the marker corresponding to the parent node is leftmost in any embedding of a skeleton.
                Hence, this marker is incident to the outer face.
                Reversing the embedding of the skeleton is equivalent to choosing the other face incident to the marker as the outer face.
                Theorem~\ref{thm:upward-planar-decomposition-tree} guarantees that any configuration of~$\mathcal T$ can be composed to an upward planar embedding.
                This means that a skeleton can be reversed if and only if both faces incident to the parent marker can be chosen as the outer face.
				This can be checked with the upward planarity test for embedded single-source graphs due to Bertolazzi et al. \cite{bertolazzi1998optimal}, which also outputs the set of faces that can be chosen as the outer face.
                If both incident faces are candidates for the outer face this node does indeed offer a degree of freedom and we leave it unchanged.
                Otherwise, if only one incident face is a candidate for the outer face this node does not offer a degree of freedom.
                We then merge it with its parent node and contract the corresponding arc in the decomposition tree.
                This leads to an R-node with a larger skeleton.

                See Fig.~\ref{fig:arc-contractions}~(a) in the appendix for an upward planar graph~$G$ and~(b) a decomposition tree thereof.
                Parts of the face sink graphs of~$\skel(\mu)$ and~$\skel(\nu)$ are shown in red, namely the two quadratic vertices dual to the faces incident to the parent marker and the edges incident to those vertices.
                One criterion for a face to be a candidate for becoming the outer face due to Bertolazzi et al.~is that there has to be a path from this face to the outer face in the face sink graph.
                This holds true for both faces incident to the parent marker in~$\skel(\nu)$, but not in~$\skel(\mu)$.
                Therefore the arc~$(\mu, \nu)$ is not contracted but the arc~$(\lambda, \mu)$ is contracted.
                This leads to the decomposition tree shown in~(c).
                See also Fig.~\ref{fig:example}~(c) and~(d) for a larger example.
                \begin{figure}[t]
                    \centering
                    \includegraphics{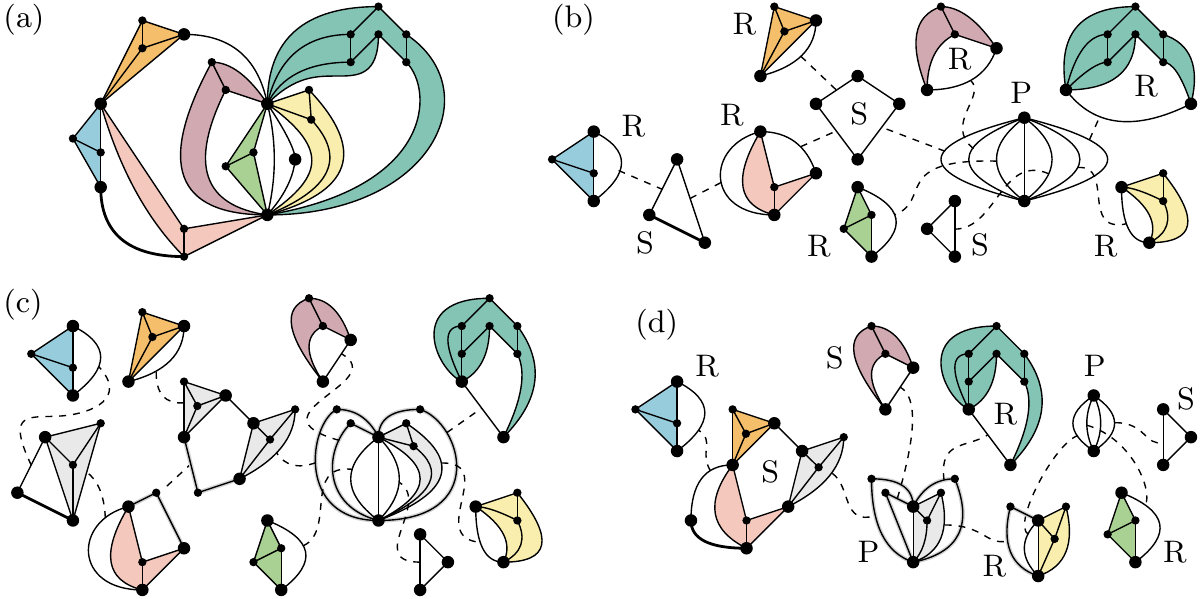}
                    \caption{
                        Construction of the UP-tree.
                        An upward planar biconnected single-source graph~(a), the SPQR-tree of its underlying undirected graph~(b) with the Q-nodes omitted, the result of replacing virtual edges with markers~(c) and the UP-tree after splitting P-nodes and contracting arcs~(d).
                    }
                    \label{fig:example}
                \end{figure}

                \begin{lemma}
                    Let~$G$ be a biconnected digraph with a single source~$s$ and let~$e^\star$ denote an edge incident to~$s$.
                    There exists a decomposition tree~$\mathcal T$ that
                    \begin{enumerate*}[label=(\roman*)]
                        \item represents all upward planar embeddings of~$G$ in which~$e^\star$ is the leftmost edge around~$s$, and
                        \item the children of all P-nodes in~$\mathcal T$ can be arbitrarily permuted.
                        \item the skeletons of all R-nodes in~$\mathcal T$ can be reversed.
                    \end{enumerate*}
                \end{lemma}

                \noindent
                We call the decomposition tree~$\mathcal{T}$ the \emph{UP-tree of~$G$ relative to~$e^\star$}.

            \subsection{Computation in Linear Time}

            Let~$G$ be a biconnected digraph with a single source~$s$ and let~$e^\star$ denote an edge incident to~$s$.
            Recall that the construction of the UP-tree~$\mathcal T$ of~$G$ relative to~$e^\star$ consists of the following seven steps.
            \begin{enumerate*}
                \item Construct the SPQR-tree~$\mathcal T$ of~$G$ in linear time~\cite{hopcroft1973dividing,gutwenger2000linear}.
                \item For each pair of vertices~$u, v$ that are the poles of a marker in some skeleton of~$\mathcal T$, we have to determine whether~$u < v$ in~$G$.
                      To compute this information for all pairs in linear time, we use a union-find-based technique described by Bl\"asius et al.~\cite{blaesius2018simultaneous}.
                      Process all skeletons of~$\mathcal T$ and for every pair of poles~$u, v$ that is encountered register~$v$ as a \emph{candidate} at~$u$ and register~$u$ as a candidate at~$v$.
                      Next, initialize every vertex of~$G$ in its own singleton set.
                      Then, process each vertex~$u$ in some reverse topological order of~$G$.
                      Unify the singleton set of~$u$ with the sets of its direct descendants in~$G$.
                      Now for any candidate~$v$ stored at~$u$ we can query in whether~$u$ and~$v$ belong to the same set, which is equivalent to~$u < v$.
                      Note that the operands to all unify operations are completely determined by the structure of~$G$.
                      We exploit this fact to run the linear-time union-find algorithm due to Gabow and Tarjan~\cite{gabow1985unionfind}.
                \item For each arc~$a = (\mu, \nu)$ of~$\mathcal T$, decide whether the poles of~$a$ are sources, sinks or internal vertices in~$G(\mu)$ and~$G(\nu)$.
                      This information can be found using a simple bottom-up technique. We first
					  compute the indegree and outdegree of every node of~$G$. We
					  then perform a depth-first traversal of~$\mathcal{T}$. We
					  maintain a list of the number of incoming and outgoing edges
					  for each node seen so far, which is updated when a Q-node is visited.
					  Upon entering a subtree, we store
					  these numbers for the poles of the arc leaving the subtree at
					  the root of the subtree.
					  Upon leaving a subtree, we can now calculate
					  the differences between the current numbers and the stored numbers,
					  which gives the in- and outdegree of the poles in the graph~$G(\mu)$. Using the in- and outdegree of the poles in~$G$ computed
					  earlier, we can also compute the in- and outdegree of the poles in~$G(\nu)$.
					  This step clearly takes linear time.
                \item In each skeleton, replace all virtual edges with their respective markers.
                      With the information that was computed in the previous step and the fact that all markers have constant size this step is feasible in linear time.
                \item Construct a configuration of~$\mathcal T$ by running the linear-time upward planar embedding algorithm of Bertolazzi et al.~\cite{bertolazzi1998optimal} on every skeleton.
                      Because the size of all skeletons is linear in the size of~$G$ this step takes linear time.
                \item Perform P-node splits.
                      The running time spent on one P-node is clearly linear in the size of its skeleton.
                      This gives linear running time overall.
                \item Perform arc contractions.
                      The upward planarity test for fixed embeddings due to Bertolazzi et al.~runs in linear time.
                      Contracting an arc is feasible in constant time.
                      This gives linear running time overall.
            \end{enumerate*}


            \begin{restatable}{theorem}{thmUPTree}
                Let~$G$ be a biconnected digraph with a single source~$s$ and let~$e^\star$ denote an edge incident to~$s$.
                The UP-tree~$\mathcal T$ of~$G$ relative to~$e^\star$ is a decomposition tree whose internal nodes are
                \begin{enumerate*}[label=(\roman*)]
                    \item S-nodes whose skeletons have a fixed embedding,
                    \item R-nodes whose skeletons have a fixed embedding up to reversal, or
                    \item P-nodes where the markers can be arbitrarily permuted in the skeleton
                \end{enumerate*}
                and whose leaves are Q-nodes that offer no embedding choice.
                The configurations of~$\mathcal T$ correspond bijectively to the upward planar embeddings of~$G$ where~$e^\star$ appears leftmost around~$s$.
                Moreover,~$\mathcal T$ can be computed in linear time.
            \end{restatable}

        \section{Partial Upward Embedding}
        \label{sec:partial-upward-embedding}

            In this section we apply the UP-tree to solve the partial upward embedding problem in quadratic time.
            A \emph{partially embedded graph} is a tuple~$(G, H, \mathcal H)$, where~$G$ is a planar graph,~$H$ is a subgraph of~$G$ and~$\mathcal H$ is a planar embedding of~$H$.
            An embedding~$\mathcal G$ of~$G$ \emph{extends} the partial embedding~$\mathcal H$ if all edges~$e, f, g$ in~$H$ that share a common endpoint~$v$ appear in the same cyclic order around~$v$ in~$\mathcal G$ and~$\mathcal H$.
            The \emph{partial embedding problem} asks whether there exists an embedding~$\mathcal G$ of~$G$ that extends~$\mathcal H$.
            Angelini et al.~solve the partial embedding problem in linear time~\cite{adfjk-tppeg-15}.
            The algorithm considers every triple of edges~$(e, f, g)$ in~$H$ that share a common endpoint~$v$ and enforces the constraints imposed by these edges in the SPQR-tree~$\mathcal T$.
            Note that~$e, f, g$ each correspond to a Q node in~$\mathcal T$.
            Because~$\mathcal T$ is a tree there is exactly one node~$\mu$ in~$\mathcal T$ so that the paths from~$\mu$ to these Q nodes are disjoint.
            The relative order of~$e, f, g$ in the embedding represented by~$\mathcal T$ is determined by the embedding of~$\skel(\mu)$.
            If~$\skel(\mu)$ offers no embedding choice (as in S nodes) determine whether the ordering of~$e, f, g$ given by~$\mathcal H$ is the same as the one given by the unique embedding of~$\skel(\mu)$.
            If not, reject the instance.
            If~$\skel(\mu)$ has two possible embeddings (as in R nodes) the ordering of~$e, f, g$ given by~$\mathcal H$ fixes one of the two embeddings of~$\skel(\mu)$ as the only candidate.
            Finally, if~$\mu$ is a P node the ordering of~$e, f, g$ given by~$\mathcal H$ restricts the set of admissible permutations of the virtual edges in~$\skel(\mu)$.
            The algorithm collects all these constraints and checks whether they can be fulfilled at the same time.

            A \emph{partially embedded upward graph} is defined as a tuple~$(G, H, \mathcal H)$, where~$G$ is an upward planar graph,~$H$ is a subgraph of~$H$ and~$\mathcal H$ is an upward planar embedding of~$H$.
            Note that UP-trees have all properties of SPQR-trees that are needed in the algorithm described above.
            In particular, the markers in P-nodes may be arbitrarily permuted, R-nodes may be reversed and all other nodes offer no embedding choice.
            Hence, we use the UP-tree as a drop-in replacement for the SPQR-tree in the algorithm of Angelini et al.~to obtain an algorithm that solves the partial upward embedding problem.
            Note that the UP-tree is rooted at some edge that must be embedded as the leftmost edge around the source of the graph.
            We may have to try a linear number of candidate edges in the worst case.
            This gives the following.

            \begin{theorem}
                The partial upward embedding problem can be solved in quadratic running time for biconnected single-source digraphs.
            \end{theorem}

    \section{Conclusion}

        We have developed the UP-tree, which is an SPQR-tree-like embedding representation for upward planarity.
        We expect that the UP-tree is a valuable tool that makes it possible to translate existing constrained planar embedding algorithms that use SPQR-trees to the upward planar setting.
        As an example, we have demonstrated how to use the UP-tree as a drop-in replacement for the SPQR-tree in the partial embedding extension problem, solving the previously open partial upward embedding extension problem for the biconnected single-source case.

    \newpage
    \appendix

    \section{Omitted Figures}

    \renewcommand\textfraction{.01}

        \begin{figure}[h]
            \centering
            \includegraphics{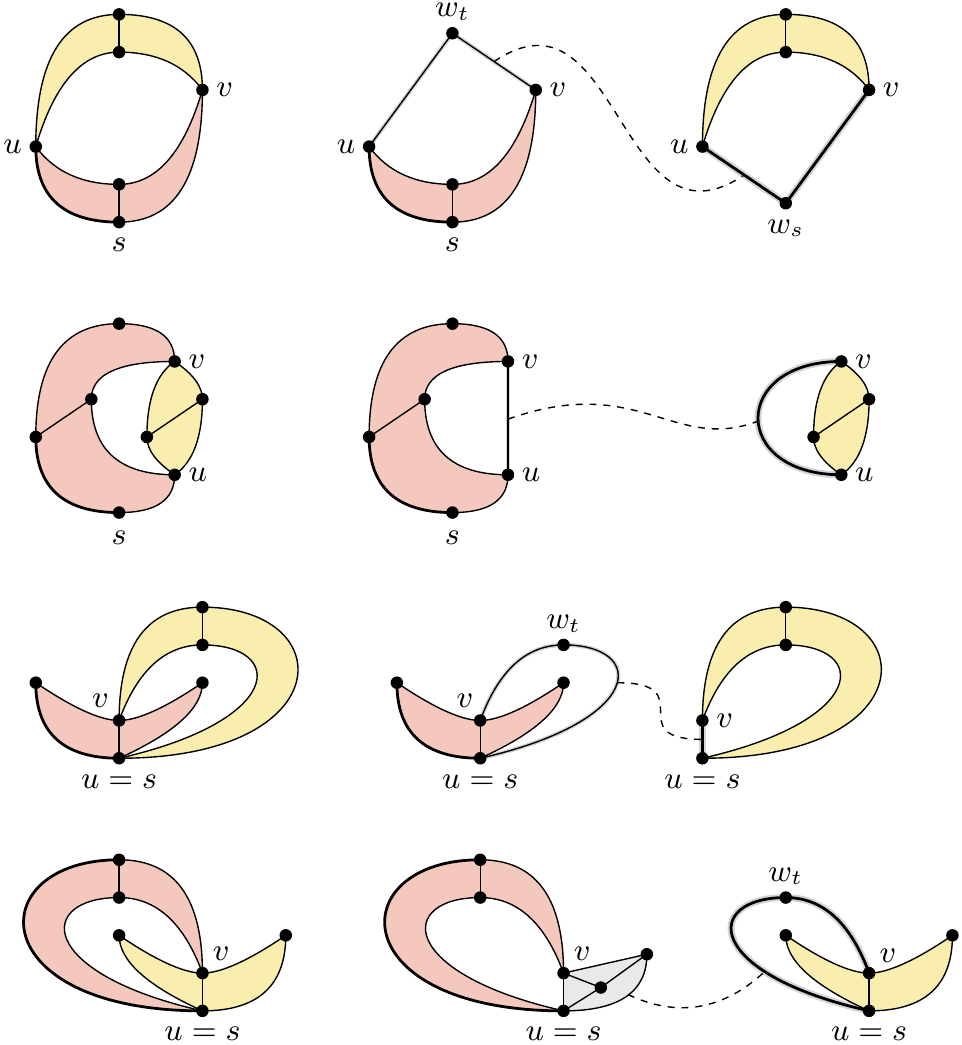}
            \caption{
                Example decompositions of upward planar graphs.
                The edge~$e^\star$ is drawn more thickly.
            }
            \label{fig:marker-replacement}
        \end{figure}

        \begin{figure}[h]
            \centering
            \includegraphics{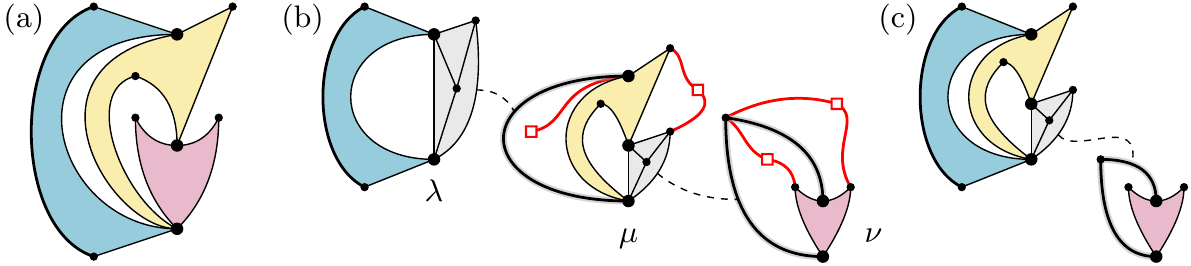}
            \caption{
                An upward planar graph~$G$~(a) with a decomposition tree~(b).
                The node~$\mu$ offers no degree of freedom so the arc~$(\lambda, \mu)$ is contracted~(c).
                The node~$\nu$ does offer a degree of freedom and is therefore not contracted.
            }
            \label{fig:arc-contractions}
        \end{figure}

    \newpage
    \section{Omitted Proofs from Section~\ref{sec:decomposition-trees-and-upward-planar-embeddings}}
    \label{sec:appendix-1}

        \theoremUpwardPlanarConfigurations*

        \begin{proof}
            We do not repeat
            all arguments of Hutton and Lubiw here, so the reader is encouraged to
            view this document and the paper by Hutton and Lubiw side-by-side.

            First, we note that the way in which Hutton and Lubiw exploit
            the fact that ~$H_1$ is a single component is their Lemma 6.2.
            However, Lemma 6.2 is completely symmetric and in particular does not
            require the isolated component to contain the source. Therefore, our
            alternative requirement of~$H_2$ being a single component is sufficient for the applications
            of Lemma 6.2 in the proofs of 6.5, 6.7, 6.8 and 6.9 in \cite{hutton1996upward}.

            It remains to see that the constructions by Hutton and Lubiw not only
            show that~$G$,~$H_1$ and~$H_2$ are upward planar, but that the upward
            planar embeddings of~$H_1$ and~$H_2$ derived from~$\mathcal{E}$ are exactly~$\mathcal{F}_1$ and~$\mathcal{F}_2$ and conversely that the upward planar
            embedding of~$G$ assembled from~$\mathcal{F}_1$ and~$\mathcal{F}_2$ is
            exactly~$\mathcal{E}$. For this, we will follow along the proof of the
            relevant theorems in \cite{hutton1996upward} and supply the necessary
            additional arguments.

            \paragraph{$u$,~$v$ incomparable.} This case is covered by Theorem 6.5
            in~\cite{hutton1996upward}. The proof of necessity finds that~$H_1'~$ and~$H_2'$ are homeomorphic to subgraphs of~$G$ and thus
            this argument shows the upward planarity of~$\mathcal{F}_1$ and~$\mathcal{F}_2$.

            The proof of sufficiency does not change any rotation system and thus
            the constructed embedding is exactly the composition of~$\mathcal{F}_1$ and~$\mathcal{F}_2$.

            \paragraph{$u<v$,~$u\neq s$.} This case is covered by Theorem 6.7
            in~\cite{hutton1996upward}. If~$v$ is not an internal vertex of~$H_2$,~$H_1'$ is homeomorphic to a subgraph of~$G$ and thus~$\mathcal{F}_1$ is
            upward planar. If~$v$ is an internal vertex of~$H_2$, then Hutton and
            Lubiw follow the upward planarity of~$H_1'$ from the fact that it has
            an upward planar subdivision that is a subgraph of~$\mathcal{E}$ and
            therefore~$\mathcal{F}_1$ must be upward planar.

            The proof of necessity of upward planarity of~$H_2'$ again finds that~$H_2'$ is homeomorphic to a subgraph of~$G$ up to some contractions which
            only happen inside the marker and thus do not change the fact that
            the upward planar embedding found for~$H_2'$ is precisely~$\mathcal{F}_2$.

            The proof of sufficiency consists of three cases. In case 1, the construction
            inserts~$H_2$ into~$H_1$ in such a way that the outgoing edges of~$u$ and~$v$ in~$H_2$ are inserted where the marker ~$M_t$ is located in~$H_1$.
            Therefore,~$\mathcal{E}$ is exactly the composition of~$\mathcal{F}_1$ and~$\mathcal{F}_2$. In case 2,~$\mathcal{F}_1$ and~$\mathcal{F}_2$ are combined
            using Lemma 5.5, which just composes~$\mathcal{F}_1$ and~$\mathcal{F}_2$ in
            our sense, yielding~$\mathcal{E}$.

            Case 3 consists of two subcases. The first case uses Lemma 5.5 and a contraction
            to obtain~$G$ and thus constructs exactly~$\mathcal{E}$ for reasons analogous to
            those in case 1. The second case again uses Lemma 5.5 and constructs the embedding
            such a way that the embeddings of~$H_1$ and~$H_2$ remain unchanged and the
            edges of~$H_2$ are inserted at the position of the marker in~$H_1$.

            \paragraph{$u=s$.} This case is covered by Theorems 6.8 and 6.9 in \cite{hutton1996upward}.
            Because Hutton and Lubiw operate under somewhat weaker assumptions compared
            to the present document, they have to rely on a specific choice of decomposition
            allowing for an explicit upward planarity test to continue decomposing.
            For the purposes of Theorem~\ref{thm:upward-planar-configurations}, this
            is insufficient. The missing argument is the necessity of the upward
            planarity of~$\mathcal{F}_2$ if~$\mathcal{E}$ is upward planar. Once this
            is established, the rest of the proof is given in Theorem 6.8.

            For this case, a simplified version of the necessity of the second
            condition of Theorem 6.7 in \cite{hutton1996upward} works. More precisely,
            we necessarily have~$z=u$ and the case where~$u$ and~$v$ are incomparable
            in~$E=H_1$ will never occur. The same arguments given above for why
            this construction yields exactly~$\mathcal{F}_1$ apply.

            For~$\mathcal{F}_1$, the proof of necessity in Theorem 6.8 applies\footnote{Note
            that Hutton and Lubiw essentially switch the roles of~$E$ and~$F$ for this proof.},
            but this proof performs the same construction as the necessity of the second
            condition in the proof of Theorem 6.7, and thus yields exactly~$\mathcal{F}_1$
            using the arguments already given above.

            The proof of sufficiency remains unchanged, and thus, once again, the
            arguments from the case where~$u\neq s$ still apply.
            \qed

        \end{proof}

        \lemmaDecompositionsPreserveVertexDominance*
        
                \begin{proof}
            Recall that the edges of~$G$ are partitioned across~$H_1$ and~$H_2$ and that~$v$ is a source in~$H_2$.
            \begin{enumerate}
                \item $u$ and~$v$ are incomparable in~$G$.
                      Let~$x, y$ be two vertices in~$H_1'$ with~$x < y$ in~$G$.
                      Because~$u$ and~$v$ are incomparable in~$G$ the directed path from~$x$ to~$y$ in~$G$ cannot use an edge from~$H_2$, i.e., it consists entirely of edges in~$H_1$, i.e., it is~$x < y$ in~$H_1'$.
                      Now let~$x, y$ be two vertices in~$H_1'$ with~$x < y$ in~$H_1'$.
                      Because~$M_t$ contains no directed path between~$u$ and~$v$ this path consists entirely of edges in~$H_1$.
                      Then this same path exists in~$G$, i.e., it is~$x < y$ in~$G$.

                      A symmetric argument using the fact that~$M_s$ contains no directed path between~$u$ and~$v$ works for the case when~$x, y$ are vertices in~$H_2'$.
                \item $u < v$ in~$G$.
                      Let~$x, y$ be two vertices in~$H_1'$.
                      Follow the construction rules for~$H_1'$.
                      \begin{enumerate}
                          \item $v$ is a source in~$H_2$.
                                Recall~$H_1' = H_1 \cup M_t$.
                                Assume~$x < y$ in~$H_1'$.
                                Because~$M_t$ has no directed path between~$u$ and~$v$ any directed path from~$x$ to~$y$ in~$H_1'$ cannot use edges from~$M_t$.
                                This means that such a path consists entirely of edges in~$H_1$, i.e., it also exists in~$G$.
                                Hence, it is~$x < y$ in~$G$.
                                Now assume~$x < y$ in~$G$.
                                Because~$v$ is a source in~$H_2$ there exists no directed path from~$u$ to~$v$ in~$H_2$.
                                From~$u < v$ in~$G$ it follows that there exists a directed path from~$u$ to~$v$ in~$H_1$ and therefore in~$H_1'$.
                                Hence, it is~$x < y$ in~$H_1'$.
                          \item $v$ is a sink in~$H_2$.
                                Because~$G$ is a single-source graph and it is~$u < v$ there exists a directed path~$p$ from~$u$ to~$v$ in~$H_2$.
                                Recall~$H_1' = H_1 \cup M_{uv}$.
                                Any directed path from~$x$ to~$y$ in~$H_1'$ either uses only edges in~$H_1$ in which case the same path exists in~$G$, or it uses the edge~$M_{uv}$ in which case it can be modified to a path that uses~$p$ in~$G$.
                                The same argument works in the reverse direction.
                          \item $v$ is an internal vertex in~$H_2$.
                                Recall~$H_1' = H_1 \cup M_{uvt}$ and apply the same argument as in the previous case.
                      \end{enumerate}

                      Now let~$x, y$ be two vertices in~$H_2'$ and follow the construction rules for~$H_2'$.
                      \begin{enumerate}
                          \item $v$ is a source in~$H_1$.
                                Recall~$H_2' = H_2 \cup M_t$ and follow the symmetric case~$H_1' = H_1 \cup M_t$.
                          \item $v$ is not a source in~$H_1$.
                                Recall~$H_2' = H_2 \cup M_{uv}$ and follow the symmetric cases~$H_1' = H_1 \cup M_{uv}$ or~$H_1' = H_1 \cup M_{uvt}$.
                    \qed
                      \end{enumerate}
            \end{enumerate}
        \end{proof}

\lemmaDecompositionsPreserveNeighborhood*

          \begin{proof}
            Distinguish the cases where~$\{x, y\}$ is a cutpair in~$H_1'$ or~$H_2'$.
            \begin{enumerate}
                \item $\{x, y\}$ is a cutpair in~$H_1'$.
                      Then~$F_1 = D_1$ and~$F_2$ is obtained from~$D_2$ by replacing~$H_2$ with the appropriate marker~$M$.
                      Equivalently,~$D_2$ is obtained from~$F_2$ by replacing the marker in~$F_2$ with~$H_2$.
                      Note that if~$M$ is not adjacent to~$y$ the neighborhood of~$y$ remains unchanged and nothing further needs to be shown.
                      Otherwise~$M$ is adjacent to~$y$, i.e., it is~$y = u$ or~$y = v$.

                      Consider the case~$y = v$.
                      If~$M = M_t$ vertex~$y = v$ is a source in~$H_2$ and so exchanging~$M = M_t$ with~$H_2$ exchanges one outgoing edge with a non-empty set of outgoing edges.
                      If~$M = M_{uv}$ vertex~$y = v$ is a sink in~$H_2$ and so exchanging~$M = M_{uv}$ with~$H_2$ exchanges one incoming edge with a non-empty set of incoming edges.
                      Finally, if~$M = M_{uvt}$ vertex~$y = v$ is an internal vertex in~$H_2$ and so exchanging~$M = M_{uvt}$ with~$H_2$ exchanges three outgoing and three incoming edges with non-empty sets of outgoing and incoming edges.

                      Now consider the case~$y = u$.
                      By definition~$u$ is a source in~$H_2$ and in all candidate markers.

                      Thus, when going from~$F_2$ to~$D_2$ and vice versa only edges of the same kind are exchanged.
                      Because~$F_1 = D_1$ vertex~$y$ is a source in~$F_1$ if and only if~$y$ is a source in~$D_1$.
                      This shows the claim for~$i = 1$.
                \item $\{x, y\}$ is a cutpair in~$H_2'$.
                      Then~$F_2 = D_2$ and~$F_1$ is obtained from~$D_1$ by replacing~$H_1$ with the appropriate marker~$M$.
                      Equivalently,~$D_1$ is obtained from~$F_1$ by replacing the marker in~$F_1$ with~$H_1$.
                      Note that if~$M$ is not adjacent to~$y$ the neighborhood of~$y$ remains unchanged and nothing further needs to be shown.
                      Otherwise~$M$ is adjacent to~$y$, i.e., it is~$y = u$ or~$y = v$.

                      Consider the case~$y = v$.
                      If~$v = y$ is a source in~$H_1$ it is~$M = M_t$ and exchanging~$M = M_t$ with~$H_1$ exchanges one outgoing edge with a non-empty set of outgoing edges.
                      Otherwise~$v = y$ is not a source in~$H_1$ and it is~$M = M_{uv}$.
                      From~$M = M_{uv}$ it follows that~$v = y$ is not a source in~$F_1$.
                      Because~$v = y$ is not a source in~$H_1$ it has at least one incoming edge~$e$ in~$H_1$.
                      Further,~$e^\star$ lies in~$D_1$ by definition and because~$H_1 \subset D_1$ edge~$e$ lies in~$D_1$, i.e.,~$v = y$ is not a source in~$D_1$.

                      The case~$y = u$ cannot occur because~$u$ is the source of~$H_2'$.
                      Because~$F_2 = D_2$ vertex~$y$ is a source in~$F_2$ if and only if~$y$ is a source in~$D_2$.
                      This shows the claim for~$i = 2$.
                      \qed
            \end{enumerate}
        \end{proof}


\begin{thebibliography}{10}
\providecommand{\url}[1]{\texttt{#1}}
\providecommand{\urlprefix}{URL }
\providecommand{\doi}[1]{https://doi.org/#1}

\bibitem{adfjk-tppeg-15}
Angelini, P., {Di Battista}, G., Frati, F., Jel{\'{\i}}nek, V.,
  Kratochv{\'{\i}}l, J., Patrignani, M., Rutter, I.: Testing planarity of
  partially embedded graphs. {ACM} Trans. Algorithms  \textbf{11}(4),
  32:1--32:42 (2015). \doi{10.1145/2629341}

\bibitem{adp-fmdep-11}
Angelini, P., {Di Battista}, G., Patrignani, M.: Finding a minimum-depth
  embedding of a planar graph in $o(n^4)$ time. Algorithmica  \textbf{60}(4),
  890--937 (2011). \doi{10.1007/s00453-009-9380-6}

\bibitem{bertolazzi1994upward}
Bertolazzi, P., Di~Battista, G., Liotta, G., Mannino, C.: Upward drawings of
  triconnected digraphs. Algorithmica  \textbf{12}(6),  476--497 (1994)

\bibitem{bertolazzi1998optimal}
Bertolazzi, P., Di~Battista, G., Mannino, C., Tamassia, R.: Optimal upward
  planarity testing of single-source digraphs. SIAM Journal on Computing
  \textbf{27}(1),  132--169 (1998)

\bibitem{bkr-sepg-14}
Bl\"asius, T., Kobourov, S.G., Rutter, I.: Simultaneous embedding of planar
  graphs. In: Tamassia, R. (ed.) Handbook of Graph Drawing and Visualization,
  pp. 349--373. Discrete Mathematics and its Applications, CRC Press (2014)

\bibitem{blr-ogdie-16}
Bl{\"{a}}sius, T., Lehmann, S., Rutter, I.: Orthogonal graph drawing with
  inflexible edges. Comput. Geom.  \textbf{55},  26--40 (2016).
  \doi{10.1016/j.comgeo.2016.03.001},
  \url{https://doi.org/10.1016/j.comgeo.2016.03.001}

\bibitem{brw-oogdc-16}
Bl{\"{a}}sius, T., Rutter, I., Wagner, D.: Optimal orthogonal graph drawing
  with convex bend costs. {ACM} Trans. Algorithms  \textbf{12}(3),  33:1--33:32
  (2016). \doi{10.1145/2838736}, \url{https://doi.org/10.1145/2838736}

\bibitem{blaesius2018simultaneous}
Bläsius, T., Karrer, A., Rutter, I.: Simultaneous embedding: Edge orderings,
  relative positions, cutvertices. Algorithmica  \textbf{80}(4),  1214--1277
  (2018)

\bibitem{br-pclp-17}
Br{\"{u}}ckner, G., Rutter, I.: Partial and constrained level planarity. In:
  Klein, P.N. (ed.) Proc. 28th Annual {ACM-SIAM} Symposium on Discrete
  Algorithms (SODA'17). pp. 2000--2011. {SIAM} (2017)

\bibitem{brs-lptec-18}
Br\"uckner, G., Rutter, I., Stumpf, P.: Level planarity: Transitivity vs. even
  crossings. In: Biedl, T., Kerren, A. (eds.) Graph Drawing and Network
  Visualization (GD'18). LNCS, vol. 11282, pp. 39--52. Springer (2018)

\bibitem{ddf-eupgd-19}
{Da Lozzo}, G., {Di Battista}, G., Frati, F.: Extending upward planar graph
  drawings. CoRR  \textbf{abs/1902.06575} (2019)

\bibitem{djkr-pesuf-14}
{Da Lozzo}, G., Jel{\'{\i}}nek, V., Kratochv{\'{\i}}l, J., Rutter, I.: Planar
  embeddings with small and uniform faces. In: Ahn, H., Shin, C. (eds.)
  Proceedings of the 25th International Symposium on Algorithms and Computation
  (ISAAC'14). Lecture Notes in Computer Science, vol.~8889, pp. 633--645.
  Springer (2014). \doi{10.1007/978-3-319-13075-0\_50}

\bibitem{dr-aafcp-18}
{Da Lozzo}, G., Rutter, I.: Approximation algorithms for facial cycles in
  planar embeddings. In: Hsu, W.L., Lee, D.T., Liao, C.S. (eds.) Proceedings ot
  the 29th International Symposium on Algorithms and Computation (ISAAC'18).
  LIPIcs, vol.~123, pp. 41:1--41:13. Schloss Dagstuhl -- Leibniz-Zentrum fuer
  Informatik (2018). \doi{10.4230/LIPIcs.ISAAC.2018.41}

\bibitem{dbt-ipt-89}
Di~Battista, G., Tamassia, R.: Incremental planarity testing. In: Proceedings
  of the 30th Annual Symposium on Foundations of Computer Science. pp. 436--441
  (Oct 1989). \doi{10.1109/SFCS.1989.63515}

\bibitem{dbt-olgawst-90}
Di~Battista, G., Tamassia, R.: On-line graph algorithms with {SPQR}-trees. In:
  Paterson, M.S. (ed.) Proceedings of the 17th International Colloquium on
  Automata, Languages and Programming. pp. 598--611. Springer Berlin Heidelberg
  (1990). \doi{10.1007/BFb0032061}

\bibitem{dt-omtcs-96}
{Di Battista}, G., Tamassia, R.: On-line maintenance of triconnected components
  with {{SPQR}}-trees. Algorithmica  \textbf{15}(4),  302--318 (1996).
  \doi{10.1007/BF01961541}

\bibitem{dlp-bmodq-18}
Didimo, W., Liotta, G., Patrignani, M.: Bend-minimum orthogonal drawings in
  quadratic time. In: Biedl, T.C., Kerren, A. (eds.) Graph Drawing and Network
  Visualization, {GD} 2018. LNCS, vol. 11282, pp. 481--494. Springer (2018).
  \doi{10.1007/978-3-030-04414-5\_34}

\bibitem{fce-pcg-95}
Feng, Q., Cohen, R.F., Eades, P.: Planarity for clustered graphs. In: Spirakis,
  P.G. (ed.) Proceedings of the 3rd Annual European Symposium on Algorithms
  (ESA'95). Lecture Notes in Computer Science, vol.~979, pp. 213--226. Springer
  (1995). \doi{10.1007/3-540-60313-1\_145}

\bibitem{fpss-htmdalp-13}
Fulek, R., Pelsmajer, M.J., Schaefer, M., {\v{S}}tefankovi{\v{c}}, D.:
  {Hanani-Tutte, Monotone Drawings, and Level-Planarity}. In: Pach, J. (ed.)
  Thirty Essays on Geometric Graph Theory. pp. 263--287. Springer New York
  (2013)

\bibitem{gabow1985unionfind}
Gabow, H.N., Tarjan, R.E.: A linear-time algorithm for a special case of
  disjoint set union. Journal of Computer and System Sciences  \textbf{30}(2),
  209--221 (1985)

\bibitem{gutwenger2000linear}
Gutwenger, C., Mutzel, P.: A linear time implementation of {SPQR}-trees. In:
  International Symposium on Graph Drawing. pp. 77--90. Springer (2000)

\bibitem{hopcroft1973dividing}
Hopcroft, J.E., Tarjan, R.E.: Dividing a graph into triconnected components.
  SIAM Journal on Computing  \textbf{2}(3),  135--158 (1973)

\bibitem{hutton1996upward}
Hutton, M.D., Lubiw, A.: Upward planar drawing of single-source acyclic
  digraphs. SIAM Journal on Computing  \textbf{25}(2),  291--311 (1996)

\bibitem{jkr-ktppe-13}
Jel\'inek, V., Kratochv\'il, J., Rutter, I.: A {Kuratowski}-type theorem for
  planarity of partially embedded graphs. Computational Geometry: Theory and
  Applications  \textbf{46}(4),  466--492 (2013)

\bibitem{jl-lpeilt-99}
J{\"u}nger, M., Leipert, S.: {Level Planar Embedding in Linear Time}. In:
  Kratochv\'{\i}l, J. (ed.) Graph Drawing and Network Visualization. pp.
  72--81. Springer Berlin Heidelberg (1999)

\bibitem{ml-ascopcg-37}
Mac~Lane, S.: A structural characterization of planar combinatorial graphs.
  Duke Mathematical Journal  \textbf{3}(3),  460--472 (1937).
  \doi{10.1215/S0012-7094-37-00336-3}

\bibitem{platt1976planar}
Platt, C.R.: Planar lattices and planar graphs. Journal of Combinatorial
  Theory, Series B  \textbf{21}(1),  30--39 (Aug 1976)

\bibitem{rsbhkmsc-asfopolg-01}
Randerath, B., Speckenmeyer, E., Boros, E., Hammer, P., Kogan, A., Makino, K.,
  Simeone, B., Cepek, O.: {A Satisfiability Formulation of Problems on Level
  Graphs}. Electronic Notes in Discrete Mathematics  \textbf{9},  269--277
  (2001), lICS 2001 Workshop on Theory and Applications of Satisfiability
  Testing (SAT 2001)

\bibitem{t-cig-66}
Tutte, W.T.: Connectivity in Graphs. University of Toronto Press (1966)

\end{thebibliography}
\end{document}